%% file: main.tex
\documentclass[5p,times,authoryear]{elsarticle}

\input{headings}
\usepackage{bigfoot}
\usepackage{tcolorbox}
\usepackage{lipsum}
\usepackage{placeins}
\defcitealias{EBA2018}{EBA, 2019} 
\defcitealias{PRA2018}{PRA, 2018} 
\defcitealias{Bas2011}{BCBS, 2011}
\defcitealias{Bas2006}{BCBS, 2006} 
\defcitealias{Bas1996}{BCBS, 1996} 
\defcitealias{EBA.NMRF}{EBA, 2020}

\makeatletter
\def\ps@pprintTitle{%
  \let\@oddhead\@empty
  \let\@evenhead\@empty
  \def\@oddfoot{\reset@font\hfil\thepage\hfil}
  \let\@evenfoot\@oddfoot
}
\makeatother

\begin{document}

\begin{frontmatter} 
\title{A novel scaling approach for unbiased adjustment of risk estimators}

\author[a1]{Marcin Pitera}
\ead{marcin.pitera@uj.edu.pl}
\address[a1]{Institute of Mathematics, Jagiellonian University,
            Krakow, Poland            }
          
\author[a2]{Thorsten Schmidt}
\address[a2]{Mathematical Institute, Albert-Ludwigs University of Freiburg,
                Freiburg, Germany}

\author[a3]{\L{}ukasz Stettner}
\address[a3]{Institute of Mathematics, Polish Academy of Sciences, 
            Warsaw, Poland
            }

\cortext[cor1]{Corresponding author}

\begin{abstract}
The assessment of risk based on historical data faces many challenges, in particular due to the limited amount of available data, lack of stationarity, and heavy tails. 
While estimation on a short-term horizon for less extreme percentiles tends to be reasonably accurate, extending it to longer time horizons or extreme percentiles poses significant difficulties. The application of theoretical risk scaling laws to address this issue has been extensively explored in the literature. 

This paper presents a novel approach to scaling a given risk estimator, ensuring that the estimated capital reserve is robust and conservatively estimates the risk. We develop a simple statistical framework that allows efficient risk scaling and has a direct link to backtesting performance. Our method allows time scaling beyond the conventional square-root-of-time rule, enables risk transfers, such as those involved in economic capital allocation, and could be used for unbiased risk estimation in small sample settings.

To demonstrate the effectiveness of our approach, we provide various examples related to the estimation of value-at-risk and expected shortfall together with a short empirical study analysing the impact of our method.
\end{abstract}


\date{First circulated: \today, This version: \today}

\end{frontmatter}

\section{Introduction}

In the financial industry, scaling of risk measures to different time horizons or different confidence levels is a well-established approach, see, e.g., \cite{Car2009}, \cite{Dan2011}, and references therein.
This procedure called \emph{risk scaling} in the following, is often  applied to value-at-risk or expected shortfall for estimating capital reserves. In particular, this includes models designed for the {\it Internal models approach} (IMA) for Pillar 1 market risk capital reporting, where the 1-day holding period is scaled to the 10-day holding period or economic capital models, where confidence level and holding period are simultaneously scaled, see \cite{EGIM,EGICAAP} for the EU regulatory background. In fact, the \emph{time-scaling procedure}, typically based on square-root-of-time rule, is the most common choice when estimating 10-day value-at-risk (VaR); this is consistently confirmed by the regulatory monitoring reports showing that around 80\% of the financial institutions (included in the studies) use time scaling when estimating 10-day VaR, see~\cite{EBA2022}.

In addition, risk scaling is frequently used to quantify risks associated with exotic risk factors within the {\it Risks not in VaR} (RNIV) and {\it Risks not in the model engines} (RNIME) regulatory frameworks. This is particularly useful when data is scarce, such as in cases where only monthly quotes are available, see~\cite{EGIM} and \cite{PRA2018} for regulatory details. Risk scaling is also an integral part of the upcoming {\it Non-modellable risk factors} (NMRF) framework, see e.g. \cite{EBA.NMRF}.

There is a rich literature on risk scaling with a particular focus on time-horizon scaling. We refer to \cite{EmbKauPat2005} for an excellent overview of general risk factor time scaling methodologies and their impact on risk estimation. Also, see \cite{BlaCaiDow2000}, \cite{DanZig2006}, \cite{BruKau2007}, \cite{WanYehChe2011}, \cite{KinWag2014}, \cite{Zha2022}, and \cite{Gua2022} for papers that directly target the risk time scaling problem, using various tools linked to GARCH modelling, quantile regression, or specific diffusion assumptions. However,  scaling based on the square-root-of-time rule leads to several problems and, in particular, to biased estimators; see \cite{DieHicInoSch1997,SaaRah2008,HamEns2010,WanYehChe2011,SkoErdChe2011,RuiHie2022} and references therein. 

There is also extensive literature on confidence level scaling, considered on a stand-alone basis or jointly with time horizon scaling. Confidence level scaling is typically based on quantile ratios, convolution approaches, semi-parametric tail estimation, extreme value theory, or dynamic scaling exponent modelling. For details, we refer to \cite{DanDeV2000}, \cite{DowBlaCai2004}, \cite{DegEmb2011}, \cite{SpaDubTer2014}, \cite{BraDiM2021}, and the references therein.

It is quite surprising that despite being a popular procedure adopted by practitioners, no industry standard for general risk scaling, especially in reference to confidence level scaling, has been established. This might be attributed to the fact that most of the developed methods are inherently linked to strong parametric assumptions, result in non-stable risk projections, produce too conservative results, or lack proper statistical error analysis.

We take this as a starting point for our work and develop a novel general framework that allows efficient risk scaling and at the same time controls for possible risk underestimation. The basis for our approach is the risk unbiasedness concept initiated in \cite{PitSch2016}. In a nutshell, we study a position secured with a scaled risk measure and determine the smallest scaling factor that renders the secured position acceptable, see Section \ref{S:robust.scaling} for precise definitions. This natural condition allows for direct control over the backtesting performance, removing unnecessary biases in other scaling methods. Both our numerical examples and the empirical study in Section~\ref{S:numerical} demonstrate the feasibility and effectiveness of our approach.

The paper is organised as follows. In Section~\ref{S:estimation} we recall the basics of risk estimation while in Section~\ref{S:scaling} we focus on the benchmark risk scaling methods. In Section~\ref{S:RiskBias} we revisit unbiasedness in the context of the estimation of risk measures adapted to our setting. The core of the paper is found in Section~\ref{S:robust.scaling} where we introduce our general framework. This is followed by examples that explain how the framework introduced in this paper could be used for efficient risk scaling in Section~\ref{S:examples}. Finally, in Section~\ref{S:numerical}, we present a simple empirical study that shows the impact of scaling on capital adequacy.

\section{Estimating risk measures}\label{S:estimation}
To lay the foundation of our work, we give a short introduction to the estimation of risk. For a detailed exposition of risk measures, we refer to \cite{FolSch2016}. Consider a reference probability space and denote by $L^0$ the space of all real-valued random variables.
Let us fix for the moment the time horizon to one day. We are interested in the risk of a position $X\in L^0$ over the time horizon. 
The (unknown) distribution of the random outcome of $X$ is denoted by $\bF_0$. For estimation, one considers a family $\cF$ of distribution functions, for example, given by a parametric family, and assumes $\bF_0 \in \cF$.

Risk is measured using a monetary risk measure $\rho$. A monetary risk measure is a mapping $\rho\colon L^0\to\bR\cup \{+\infty\}$ which satisfies \emph{monotonicity}, i.e.~$X \le Y$ implies $\rho(X) \ge \rho(Y)$, and \emph{cash invariance}, i.e.~for $m \in \bR$, $\rho(X+m) = \rho(X)-m$. Furthermore, we will assume that $\rho$  is \emph{positively homogeneous}, such that for all $\lambda \ge 0$, $\rho(\lambda X)= \lambda \rho(X)$ and that $\rho$ is \emph{law-invariant}, i.e.\ there is a function $R:\cF \to \bR \cup \{+\infty\}$ such that $\rho(X)=R(\bF)$ whenever $X \sim \bF$.

For estimation, there is a sample $\boldsymbol{X}:=(X_1,\ldots,X_n)$ at hand. We assume that $X_1,\dots,X_n$ are i.i.d.~with $X_1 \sim \bF_0$ and that the sample is independent of $X$. An estimator of $\rho(X)$ is simply a measurable function from the sample to the real numbers, which we denote by $\hat\rho_n\colon \bR^n\to \bR$. To simplify the exposition, we make the technical assumption $\rho(X)\geq -\bE[X]$ and sometimes use $\hat\rho$ instead of $\hat\rho_n$.

With a slight abuse of notation, throughout this paper, the subscript notation in $\hat\rho$ (and $\rho$) will depend on the underlying context. It might be used to emphasise the underlying sample size, risk measure confidence threshold, or the holding period specification -- we hope this will be clear from the context. 
 
As for the risk measure estimation, we require the following properties from the estimator. Note that monotonicity is not required, as it might be not satisfied by parametric risk estimators.

\begin{definition}[Risk estimator]\label{def:risk.estimator}
The measurable function $\hat\rho_n\colon \bR^n\to \bR$ is called a {\it risk estimator} if it satisfies 
\begin{enumerate}[1)]

\item {\it cash invariance}, i.e. for any $\boldsymbol{x}\in\bR^n$ and $m\in\bR$ it holds that $\hat\rho_n(\boldsymbol{x}+m)=\hat\rho_n(\boldsymbol{x})-m$;
\item {\it positive homogenity}, i.e. for any $\boldsymbol{x}\in\bR^n$ and $\lambda \ge 0$ it holds that $\hat\rho_n(\lambda\boldsymbol{x})=\lambda\hat\rho_n(\boldsymbol{x})$.
\end{enumerate}
\end{definition}

Such a property relates to so-called {\it equivariant estimators} in a statistical context, see Chapter 10 in \cite{keener2010theoretical} for example. Hence, we require that risk estimators inherit two axiomatic properties of the underlying risk measure $\rho$. Since we are dealing mostly with {\it value-at-risk} (VaR) and {\it expected shortfall} (ES) risk measures, we assume positive homogeneity; this property could be dropped if one is interested in estimating convex risk measures. 

While the properties in Definition~\ref{def:risk.estimator} are reasonable requirements, they are obviously not sufficient for correct estimation and further properties linked to unbiasedness, consistency, or efficiency are necessary. Since the main focus of this paper is scaling, we will  not discuss estimation procedures in detail and simply start from a given  estimator. We refer to \cite{McnFreEmb2010}, \cite{KraSchZah2014},  \cite{BarKolDij2023}, and references therein for more information about risk estimation. 

Let us consider some examples: on the one hand, for estimating VaR at level 1\% and $n=250$, one could use the non-parametric quantile estimator given by
\begin{equation}\label{eq:varHS}
\hat\var^1(\boldsymbol{X}):=-\tfrac{1}{2}(X_{(2)}+X_{(3)}),
\end{equation}
where $X_{(k)}$ is the $k$th order statistic of the sample, or the parametric estimator given by
\begin{equation}\label{eq:varNORM}
\hat\var^2(\boldsymbol{X}):=-(\hat\mu(\boldsymbol{X})+\hat\sigma(\boldsymbol{X})\Phi^{-1}(0.01)),
\end{equation}
where $\Phi$ is the standard normal cumulative distribution function, $\hat\mu(\boldsymbol{X})$ is the sample mean, and $\hat\sigma(\boldsymbol{X})$ is the sample standard deviation. 

On the other hand, for estimating ES at level $\alpha=2.5\%$ and $n=250$, one could use the non-parametric estimator 
\begin{equation}\label{eq:esHS}
\hat{\textrm{ES}}^1(\boldsymbol{X}):=-\frac{1}{6}\sum_{i=1}^{6}X_{(i)},
\end{equation}
or the parametric  estimator
\begin{equation}\label{eq:esNORM}
\hat{\textrm{ES}}^2(\boldsymbol{X}):=-\left(\hat\mu(\boldsymbol{X})+\hat\sigma(\boldsymbol{X})\frac{-\phi(\Phi^{-1}(0.025))}{0.975}\right),
\end{equation}
see \cite{McN1999} and \cite{Car2009} for details.

Most of the risk estimators considered in the literature, including the ones above, are {\it plug-in estimators}. 
These are obtained by estimating the underlying distribution by a classical estimation methodology and then plugging the estimated distribution function into the calculation of the risk measure:  if the law-invariant risk measure is $\rho(X)=R(\bF_0)$ and $\bF_0$ is estimated by $\hat F_n$, then the plug-in estimator is given by $R(\hat F_n)$. 
For example, estimators \eqref{eq:varHS} and \eqref{eq:esHS} are obtained by plugging in the empirical distribution function while \eqref{eq:varNORM} and \eqref{eq:esNORM} are obtained by plugging in the normal distribution with mean estimated by $\hat\mu(\boldsymbol{X})$ and standard deviation estimated by $\hat\sigma(\boldsymbol{X})$. 

\section{Benchmark plug-in scaling methods}\label{S:scaling}

Two methods are the most common scaling methods in industry practice: for time scaling, e.g. from a 1-day holding period to a 10-day holding period, one uses the {\it square-root-of-time rule} which will be explained in detail in Section \ref{S:square.root}. For confidence level scaling, e.g. from 1.00\% to 0.05\% confidence threshold, one uses the {\it theoretical quantile ratio} which we present in Section~\ref{S: risk ratios}. 
Both methods are obtained by computing a scalar for a theoretical scenario and using this scalar for scaling the risk estimator. In fact, most scaling methods considered in the literature are obtained by a similar procedure.

Those two methods are the most common choice for the following two reasons: simplistic nature and avoidance  of model-induced risk that might lead to non-stable outcomes. Many alternative methods in the literature are appealing from a theoretical point of view and lead to better top-level performance. However, to the best of our knowledge, they are not frequently used in  production environments to avoid internal framework incoherence. 

To explain this, consider a single risk factor GARCH-type scaling or a univariate diffusion-based scaling (see \cite{EmbKauPat2005} for the overview of single risk factor scaling methods). In such a setting, the sample of {\it profits and losses} (P\&Ls) on a portfolio level is often constructed as a sum of P\&Ls of individual positions.
Those individual positions depend on a high number of risk factors that must be shocked jointly. Consequently, the construction of a consistent portfolio-level P\&Ls relies on a high-dimension sample construction method which should not be subject to further top-level dynamic adjustments, see \cite{Jor2007} for details. Having this in mind, we decided to benchmark our methodology with the aforementioned basic scaling methods, sometimes reinforced with further distributional assumptions. For a summary of other approaches, we refer to the introduction, where methods based on quantile regression, specific diffusion assumptions, GARCH dynamics, extreme value theory, or convolution methods are mentioned together with literature references.

\subsection{Time-scaling: square-root-of-time rule}\label{S:square.root}

The square-root-of-time rule is a common tool for shifting the risk from one holding period to another holding period. 
While the method is a good approximation under the assumption of i.i.d.~and normally distributed returns, these are often violated in financial data, see  \cite{WanYehChe2011} for a detailed discussion and empirical assessment of the consequences. 
Still, this method is the most common holding periods scaling tool, especially when a 1-d holding period is scaled to a 10-day holding period, and begins with a short exposition of the rule. We refer to  \cite{WanYehChe2011}, \cite{DanZig2006}, and \cite{EBA2022} for literature and further details.

For simplicity, let us start from a given 1-d risk estimator $\hat\rho$ and aim at transferring it to a holding period of $m$ days,  for $m\in\bN$. 
The scaled risk estimator $\hat\rho_m$ under the  \emph{square-root-of-time rule} is given by 
\begin{equation}\label{eq:srqt}
\hat\rho_m(\boldsymbol{X}) = \sqrt{m} \cdot \hat\rho(\boldsymbol{X}).
\end{equation}
This approach is most frequently associated with a zero-mean Gaussian setup:  let  $Z_i\sim\cN(0,\sigma)$, $i=1,\ldots,m$, denote 1-day P\&Ls with unknown  $\sigma>0$. Then, the aggregated $Z:=Z_1+\ldots+Z_m$ refers to an $m$-day P\&L. If $(Z_i)_{i=1}^{m}$ are independent,  $Z \sim \cN(0,\sqrt{m}\cdot \sigma)$. Then, for a monetary and positively homogeneous risk measure $\rho$, it holds that
\[
\rho(Z)=\rho\left(\sqrt{m}\cdot Z_1\right)=\sqrt{m}\cdot \rho(Z_1).
\]
In theory, this method already fails when one of the assumptions is dropped: i.i.d., vanishing mean or stability of the distribution  under summation.

For a short holding period scaling (e.g. when $m=10$) typically the mean $\mu$ is much smaller in comparison to the standard deviation $\sigma$, i.e.~$\mu\ll \sigma$. Also, the standard deviation of $Z$ is proportional to (or dominated by) $\sqrt{m}\cdot \sigma$. This might partly justify the popularity of \eqref{eq:srqt} within IMA methodologies. 

For a longer time horizon,  one needs to take into account that the mean scales linearly in time: again, in the above setting, \begin{equation}\label{eq:mean.scale.normal}
\rho(Z)=-m\cdot\mu+\sqrt{m}\cdot(\rho(Z_1)+
\mu),
\end{equation}
and this formula might be used for statistical estimation with $\mu$ replaced with a sample estimator $\hat\mu$ and $\rho(Z_1)$ replaced with $\hat\rho(\boldsymbol{X})$, see \cite{Car2009} for details.

The square-root-of-time rule can also be used for \emph{down scaling}, for example when scaling a 10-day risk to a 1-day risk. In this case, one  simply scales with  $m=1/\sqrt{10}$.

\subsection{Confidence-scaling: normal risk ratios}\label{S: risk ratios}

Scaling the confidence level is typically used when transferring risk to a more extreme confidence level, e.g. from 1.00\% to 0.05\% in the context of regulatory measures. Let $\rho_\alpha$ denote VaR or ES at level $\alpha\in (0,1)$ with corresponding estimator $\hat \rho_{\alpha}(\boldsymbol{X})$. Scaling from level $\alpha$ to level $\beta\in (0,1)$ based on the \emph{normal risk ratio}, is obtained by multiplying with the corresponding factor under normality, the normal risk ratio $d$ given by 
\[
d(\alpha,\beta):=\frac{\rho_{\alpha}(Z)}{\rho_{\beta}(Z)},\quad\textrm{$Z\sim N(0,1)$},
\]
such that the scaled estimator $\hat\rho_{\beta}(\boldsymbol{X})$ equals
\[
\hat\rho_{\beta}(\boldsymbol{X})=d(\alpha,\beta)\cdot \hat\rho_{\alpha}(\boldsymbol{X}).
\]
For VaR the normal risk ratio computes to  
\[
d_1(\alpha,\beta):=\frac{\Phi^{-1}(\alpha)}{\Phi^{-1}(\beta)},
\]
see \cite{Jor2007} and \cite{SpaDubTer2014}. As for the square-root-of-time rule, this procedure is inherently linked to a zero-mean-normality assumption: indeed, assuming $Z\sim \cN(0,\sigma)$, 
\[
\rho_{\alpha}(Z)=\sigma\Phi^{-1}(\alpha)=\tfrac{\Phi^{-1}(\alpha)}{\Phi^{-1}(\beta)}\cdot \sigma\Phi^{-1}(\beta)=d(\alpha,\beta)\cdot \rho_{\beta}(Z),
\]
which explains why the VaR confidence-level scaling is based on the quantile ratio and shows that $\hat\rho_{\beta}(\boldsymbol{X})$ is a simple plug-in estimator under the zero mean normality assumption. Note that for ES, using \eqref{eq:esNORM}, we obtain
\[
d_2(\alpha,\beta):=\frac{\phi(\Phi^{-1}(\alpha))}{\phi(\Phi^{-1}(\beta))}\cdot \frac{1-\beta}{1-\alpha}.
\]
As before, this estimator can be easily adjusted to account for a non-zero mean as in Equation \eqref{eq:mean.scale.normal}. Also, it is possible to work under different distributional assumptions leading to  other confidence-level transfer constants. That said, the confidence level scaling is often accompanied by time scaling, and the normal quantiles are often justified by the central limit theorem.

\section{A small introduction to unbiased risk estimation}\label{S:RiskBias}
In this section we shortly revisit the notion of {\it risk bias}, and show that ensuring unbiased  risk estimation improves backtesting performance. We refer to~\cite{PitSch2016} for a detailed exposition and empirical results.  Recall that a risk estimator is a measurable function $\hat\rho_n\colon \bR^n\to \bR$ of the given sample $\boldsymbol{X}$.
We distinguish an observed value of the sample $\mathbf{x} \in \bR^n$ from the random variable $\boldsymbol{X}$.
By the i.i.d.\ assumption the distribution of $\boldsymbol{X}$ is given by the product measure of the (unknown) reference distribution $\bF_0$ of the risky position $X$. 
Note that the estimator $\hat\rho_n(\boldsymbol{X})$
is again a random variable, highlighting that the outcome of the estimation is random and hence varies from sample to sample. Unbiasedness in a statistical context requires that the expectation of the estimator equals the estimated quantity, which we recall in the following.
\begin{definition}[Statistical unbiasedness]
The estimator $\hat \rho_n$ is  {\it statistically unbiased} for $\rho(X)$ if 
\begin{equation}\label{eq:statistical.bias}
\bE[\hat\rho_n(\boldsymbol{X})]=\rho(X),
\end{equation}
under the (unknown) reference  probability measure.
\end{definition}
Since the reference probability measure is  not known, one requires property \eqref{eq:statistical.bias} for all distributions in the considered family $\cF$.

As observed in~\cite{PitSch2016}, statistical unbiasedness is not well-suited for risk management purposes: the expectation is linked to the average of many repeated measurements and the statistical unbiasedness property implies that under- and over-estimation balances on the long run. However, in risk management one prefers capital adequacy in terms of risk over properties in average. The concept of \emph{risk unbiasedness} remedies this, as we illustrate in the following.
The key is to consider the 
{\it secured position} 
\begin{align}
    \label{eq:secured position}
    S &= X+\hat\rho_n(\boldsymbol{X})
\end{align}
which is obtained by adding the estimated capital reserve to the underlying position $X$. The main goal of risk management is to achieve acceptability of the secured  position $S$, i.e.\ the secured position should be free of unwanted risk, corresponding to an adequate estimation of the risk capital. This is precisely what is required from a risk unbiased estimator, which we define in the following.

\begin{definition}[Risk unbiasedness]
The estimator $\hat\rho_n$ is   {\it risk unbiased} for $\rho(X)$ if 
\begin{equation}\label{eq:risk.bias}
\rho(X+\hat\rho_n(\boldsymbol{X}))=0.
\end{equation}
under the (unknown) reference probability measure.
\end{definition}
Again, since the reference probability measure is unknown, we require this property for all distributions in $\cF$.

The above equation mimics key properties of monetary risk measures: $\rho(X)$ is the adequate capital reserve
to render $X$ an acceptable position, i.e.\ if we knew the risk of $X$, we could simply use cash invariance to obtain
\[
\rho(X+\rho(X))=\rho(X)-\rho(X)=0.
\]
Unfortunately, since the distribution of $X$ is unknown, it has to be estimated by the random variable $\hat\rho(\boldsymbol{X})$ depending on the particular realisation of the historical data. 
The relation to statistical bias is obtained when the risk measure under consideration is the expectation, since then $\rho(\cdot)=-\bE[\cdot]$ so that  \eqref{eq:statistical.bias} and \eqref{eq:risk.bias} coincide.
Except for this case, the concept of risk unbiasedness is preferable over statistical unbiasedness for risk estimators.

\begin{remark}[Risk unbiasedness as optimality condition]
We want to emphasize that risk unbiasedness  should not be used as a single optimality criterion when looking for the best risk estimator (in the same way that unbiasedness should not be used as a sole criterion in the classical statistical setup):  while  risk unbiasedness guarantees proper capital securitisation, it permits the usage of estimators in which the reserves are too high, see Example 7.3 in \cite{PitSch2016} for further details. 
\end{remark}

\subsection{The link to backtesting}\label{S:backtesting}
Backtesting is the most common performance or conservativeness evaluation tool for risk measures, see \cite{Acerbi2014Risk}, \cite{HeKouPen2022}, \cite{DuPeiWanYan2023}, and references therein.

As shown in \cite{MolPit2017}, risk unbiasedness has a direct link to backtesting performance for VaR when following the regulatory backtesting framework, see \cite{Bas1996}. The link could be recovered by considering a performance measure that is dual to the family of VaR risk measures. Namely, given the family of risk measures $\{\var_\alpha\}_{\alpha\in (0,1)}$ that is decreasing with respect to the confidence threshold $\alpha$ and following the generic dual performance measure framework developed in \cite{CheMad2009}, we can define a {\it performance measure} $T$ (also called {\it acceptability index}), dual to the risk family $\{\var_\alpha\}_{\alpha\in (0,1)}$, by setting
\begin{equation}\label{eq:T.perf}
T(\cdot)=\inf\{\alpha\in (0,1): \var_{\alpha}(\cdot)\leq 0\}.
\end{equation}
In particular, for the secured position $S$, the value $T(S)$ identifies the smallest confidence level $\alpha\in (0,1)$ under which the secured position $S$ is acceptable. Of course, given the VaR reference threshold $\alpha_0\in (0,1)$, we want the value of $T(S)$ to be close to $\alpha_0$.

Now, if we consider a secured position sample and an empirical equivalent of \eqref{eq:T.perf} then we can recover the standard {\it exception rate} statistic that is used for regulatory VaR backtesting.
Namely, let us assume that we are given a secured position sample vector $\boldsymbol{S}=(S_1,\ldots,S_n)$ of size $n\in\bN$, that could be obtained e.g. by summing daily risk projections with realised portfolio P\&Ls. We define the empirical version of \eqref{eq:T.perf} given by
\begin{equation}\label{eq:T.perf2}
\hat T(\boldsymbol{S})=\inf\{\alpha\in (0,1): -S_{(\lfloor n\alpha\rfloor +1)}\leq 0\};
\end{equation}
note that \eqref{eq:T.perf2} is constructed by replacing theoretical VaR measure with a non-parametric empirical VaR estimator. Then, from Proposition 3.1 in \cite{MolPit2017}, we get
\[
\hat T(\boldsymbol{S})=\sum_{i=1}^{n}\frac{\1_{\{S_i<0\}}}{n}.
\]
This shows that $\hat T(\boldsymbol{S})$ is effectively counting how many times (on average) the secured position was non-positive, i.e. the risk reserve was not sufficient to cover the realised P\&L losses. This performance measure is in fact an averaged version of the {\it breach count} statistic defined in \cite{Bas1996}. We refer to \cite{MolPit2017} and \cite{PitSch2022} for more details and proofs.

Now, let us assume that the risk estimator is risk unbiased for VaR at reference level $\alpha_0\in (0,1)$, corresponding to  $\var_{\alpha_0}(S)=0$.
Then, recalling that the family $\{\var_\alpha\}_{\alpha\in (0,1)}$ is decreasing with respect to $\alpha$, we immediately get $T(S)=\alpha_0$. Consequently, we expect the empirical value $\hat T(\boldsymbol{S})$ to be also close to $\alpha_0$. In other words, risk unbiased estimators should behave well in reference to regulatory backtesting.

\section{Risk unbiased scaling}\label{S:robust.scaling}
In this section, we introduce a new scaling methodology that achieves an appropriate notion of risk unbiasedness. 
For the beginning, we assume vanishing mean of $X$ (or $\mu\ll \sigma$) while a natural extension to the general non-vanishing mean setup detailed in Remark~\ref{rem:non-zero.mean}.

Assume that we are given a reference risk estimator $\hat\rho_n$ based on the sample $\boldsymbol{X}$.
We call $c \cdot \hat \rho_n(\boldsymbol{X})$ with $c \ge 0$ a \emph{scaled risk estimator} and introduce, 
in analogy to Equation \eqref{eq:secured position},  the \emph{scaled secured position}
\[
S(c):=X+c\cdot \hat\rho_n(\boldsymbol{X}), \quad c \ge 0.
\]
We recover the secured position by $S=S(1)$. Typically, the family $(S(c))_{c\ge 0}$ will be non-decreasing in $c$. Indeed, this is already the case when $\hat\rho_n(\boldsymbol{X})\geq 0$. We point out  that in typical situations  $\rho(X)> 0$: for example, in the value-at-risk case with mean zero, if the confidence level $\alpha$ is not too large, the value-at-risk will be positive. In the Gaussian case, this is true for $\alpha<0.5$.  A large positive mean can lead to a negative risk measure, compare Equation \eqref{eq:varNORM}.

Consequently, under some mild conditions imposed on $\rho$, such as properness, we will be able to find a minimal constant $c>0$ such that the position $S(c)$ is risk unbiased. 
To this end, we set
\begin{equation}\label{eq:c.star}
c^*_{\bF_0}:=\inf\{c>0\colon \rho(S(c))\leq 0\},
\end{equation}
with convention $\inf\emptyset=\infty$ and define the optimally scaled estimator
\begin{equation}\label{eq:rho.star}
\hat\rho_n^0(\boldsymbol{X}) := c^*_{\bF_0} \cdot \hat\rho_n(\boldsymbol{X}).
\end{equation}

The intuition behind \eqref{eq:c.star} is the following: we seek the smallest scaling factor $c$  such that the scaled position $S(c)$ carries no risk. We emphasize that $\rho$ depends on the (unknown) reference probability $\bF_0$ and will show in the following how to overcome this dependence.
The scaling idea introduced in \eqref{eq:c.star} can be used for various contexts, including general scaling, time-scaling scaling, confidence-threshold scaling, or small-sample scaling. We will show in Equation \eqref{eq:scaled estimator unbiased} below that the optimally scaled estimator is unbiased under weak assumptions.

Computationally, the main challenge in computing $c_{\bF_0}$ is linked to the fact we need  to compute $\rho(S(\cdot))$, which a priori requires information about the underlying distribution, to obtain the optimal choice $c^*_{\bF_0}$. This is a common statistical challenge shared by many frameworks. In particular, we need the same information to compute statistical bias or risk bias evaluation -- see \eqref{eq:statistical.bias}~and~\eqref{eq:risk.bias}.

To remediate the problem with the dependence on $\bF_0$, we suggest the following robust approach: take the most conservative scalar value for the pre-specified family $\cF$, i.e.\ choose 
\begin{equation}\label{eq:c.star3}
c^* := \sup_{\bF\in\cF}c^*_{\bF},
\end{equation}
where $c^*_{\bF}:=\inf\{c>0\colon \rho_{\bF}(S(c))\geq 0\}$ and $\rho_{\bF}(X):=R(\bF)$ is the risk measure when $X\sim \bF$ (this is possible due to law-invariance of $\rho$). The associated \emph{robust scaled estimator} is denoted by 
\begin{equation}\label{eq:rho.star most conservative}
    \hat\rho_n^*(\boldsymbol{X}) := c^* \cdot \hat\rho_n(\boldsymbol{X}).
\end{equation}

Since the robust estimator is the most conservative one, it might be biased  for  the (unknown) $\bF_0\in\cF$.
However, it does not underestimate the risk, i.e.\ it is biased in the right direction as illustrated in the following result.

\begin{proposition}\label{pr:unb}
Assume that  $\hat\rho_n(\boldsymbol{X})$ is bounded, non-negative, and we have $c^*<\infty$. Then,
the secured position scaled with $c^*$ is riskless, that is
\begin{equation}\label{eq:bias.negative}
    \rho(S(c^*))= \rho(X+\hat\rho_n^*(\boldsymbol{X}))\leq  0
\end{equation}
under the reference probability measure $\bF_0\in \cF$. Moreover, $\hat\rho_n^*(\boldsymbol{X})$ is the smallest scaled risk estimator such that \eqref{eq:bias.negative} holds for all $\bF \in \cF$.
\end{proposition}
\begin{proof}
Observe that $\infty>c^*\geq c^*_{\bF_0}>0$ and let
\begin{equation}\label{eq:z.fun}
Z(c):=\rho(X+c \cdot \hat\rho_n^*(\boldsymbol{X})).
\end{equation}
Since $\rho$ is monotone and $\hat\rho_n(\boldsymbol{X}) \ge 0$, we know that $Z$ is non-increasing. Moreover, since $\hat\rho_n(\boldsymbol{X})$ is bounded, $Z$ is continuous, as for any $\epsilon>0$ we get 
\[
Z(c) -\epsilon \|\hat\rho_n(\boldsymbol{X})\|_{\textrm{sup}}\leq Z(c\pm \epsilon)\leq Z(c) +\epsilon \|\hat\rho_n(\boldsymbol{X})\|_{\textrm{sup}}.
\]
Consequently, directly from~\eqref{eq:c.star}, we get
   \begin{align} \label{eq:scaled estimator unbiased}
         \rho(X+c^*_{\bF_0} \cdot \hat\rho_n^*(\boldsymbol{X})) =  0,
    \end{align}
Thus, to conclude the proof of~\eqref{eq:bias.negative}, it is sufficient to note that $c^* \ge c^*_{\bF_0}$ and recall monotonicty of $Z$.

For the second claim consider $\gamma < c^*$. Since $\gamma < \sup_{\bF\in\cF}c^*_{\bF}$, there exists $\bF \in \cF$ such that $\gamma < c^*_{\bF} \le c^*$. Moreover, by~\eqref{eq:c.star}, $c^*_{\bF}$ is the smallest constant $c$ such that $\rho(S(c))\le 0$ under $\bF$. Since $\gamma <c^*_{\bF} $ we get $\rho(S(\gamma))> 0$, at least for $\bF$, and hence the second claim follows.   
\end{proof}
To avoid technical exposition, in Proposition~\ref{pr:unb} we assumed that the risk estimator is bounded. This assumption could be easily relaxed as effectively we only need continuity of function $Z$ defined in~\eqref{eq:z.fun}. This can be achieved by requiring weak-convergence type continuity from $\rho$ (or $R$) on a space containing the family $\{S(c)\}_{c\in \bR}$. 

Let us now provide a series of further remarks linked to our methodological proposal.  In Section~\ref{S:examples} we collect further  examples  showing how our setup could be used in practical situations.

\begin{remark}[Non-zero mean]\label{rem:non-zero.mean}
If the condition $\mu\ll \sigma$ does not hold, then the scalar introduced in \eqref{eq:c.star} might lead to a non-accurate risk estimation due to the uncertainty encoded in the location parameter $\mu$ and the fact that risk should be scaled in linear proportion to $\mu$, see \eqref{eq:mean.scale.normal}. That saying, since monetary risk measures are cash additive and most estimators considered in the literature are based on statistics linear with respect to the location parameter, we can consider a modified scaled estimator given by
\begin{equation}\label{eq:rho.star2}
\hat\rho^*(\boldsymbol{X})=-\hat\mu(\boldsymbol{X})+c^* \cdot (\hat\rho(\boldsymbol{X})+\hat\mu(\boldsymbol{X})),
\end{equation}
in which $c^*$ is equal to 
\begin{equation}\label{eq:c.star2}
c^*=\min\{c>0\colon \rho(X-\hat\mu(\boldsymbol{X})+c \cdot (\hat\rho(\boldsymbol{X})+\hat\mu(\boldsymbol{X})))\geq 0\}.
\end{equation}
In a nutshell, we simply center our position to have zero mean, and then apply scaling, similarly as done in \eqref{eq:mean.scale.normal}. Note that this approach allows a straightforward extension of our framework to the situation in which the condition $\mu\ll \sigma$ does not hold.
\end{remark}

\begin{remark}[Scale invariance]
Note that if both $\rho$ and $\hat\rho$ are positively homogeneous, the value \eqref{eq:c.star3} is invariant with respect to scale parameters, i.e. for any $\bF$, $a>0$, and rescaled values (aX,$a\boldsymbol{X}$), we have
\[
\rho_\bF(aX+c_{\bF}^*\cdot \hat\rho(a\boldsymbol{X}))=a\cdot\rho_\bF(X+c_{\bF}^*\cdot \hat\rho(\boldsymbol{X}))\leq 0.
\]
\end{remark}

\begin{remark}[Scalar value as performance measure]
The scalar value $c^*$ could be seen as a  performance measure defined in \cite{CheMad2009}. Given an initial position $X$, we can consider the family of increasing risk measures 
\[
\rho^{c}(Z):=\rho(X+c\cdot Z),
\]
defined for any random variable $Z\in \cZ$, where $\cZ$ corresponds to the space of all non-negative risk estimators. In this setting, the index dual to the family $(\rho^{c})_{c>0}$ is given by
\[
\alpha(Z):=\inf\{c>0\colon \rho^{c}(Z)\leq 0\}.
\]
In particular, in this setting we get
\[
\alpha(\hat\rho(\boldsymbol{X}))=c_{\bF_0}^*,
\]
which shows that $c_{\bF_0}^*$ might be seen as a value of the acceptability index for the unscaled estimator $\hat\rho$. This provides further motivation behind the definition of the scalar $c_{\bF_0}^*$ introduced in \eqref{eq:c.star}.
\end{remark}

\begin{remark}[Asymptotic scalar behaviour]\label{rem:asymptotic}
If $n\to\infty$, we expect the value of the scalar to go to 1, at least if the underlying estimator is consistent. In particular, this result should hold in the class of {\it plug-in estimators} under the assumption that the distribution estimation process is effective. For example, this refers to the situation when Glivenko-Cantelli Theorem can be used, see \cite{Van2000}. This is the case if the underlying risk estimator exhibits asymptotic risk unbiasedness property, see Section 6 in \cite{PitSch2016} for details.
\end{remark}

\bigskip

\section{Examples}\label{S:examples}
In this section, we show some examples that illustrate how to apply the risk scaling methodology introduced in Section~\ref{S:robust.scaling}. In Section~\ref{S:general.scaling} we present generic statistical examples, while in Section~\ref{S:specific.scaling} we focus on common practical situations linked to existing regulatory frameworks.

\subsection{General scaling}\label{S:general.scaling}
The scalar introduced in the previous section could be used to refine various risk estimators. In this section, we show four simple examples of how this could be achieved. For other examples, we refer to \cite{PitSch2022}, \cite{BigTsa2016}, and \cite{GerTsa2011} where an indirect scaling based on risk unbiasedness idea was studied for specific risk measures and parametric settings; the examples presented therein could be easily adapted to our setting.

\begin{example}[Parametric 1-day VaR estimator under normality]\label{ex:1}
Let us assume that $X\sim \cN(\mu,\sigma)$, for unknown parameters $\mu\in\bR$ and $\sigma>0$. Moreover, let us fix the risk measure $\rho=\var_{\alpha}$ for a pre-defined confidence threshold $\alpha\in (0,1)$. In this setting, as pointed out it \eqref{eq:varNORM}, it is a common choice to use the standard Normal plug-in VaR estimator given by
\[
\hat\rho(\boldsymbol{X})=\hat\var_{\alpha}^2(\boldsymbol{X}):=-\left(\hat\mu(\boldsymbol{X})+\hat\sigma(\boldsymbol{X})\Phi^{-1}(\alpha)\right),
\]
Quite surprisingly, as shown in \cite{PitSch2016}, the estimator $\hat\var_{\alpha}^2$ is risk biased and there exists an unbiased alternative given by
\[
\hat\var_{\alpha}^3(\boldsymbol{X}):=-\bigg(\hat\mu(\boldsymbol{X})+\hat\sigma(\boldsymbol{X})\sqrt{\frac{n+1}{n}}t_{n-1}^{-1}(\alpha)\bigg),
\]
where $t_{n-1}$ corresponds to the cumulative distribution function of a student-$t$ distribution with $n-1$ degrees of freedom.

The proposed scaling framework  allows to recover $\hat\var_{\alpha}^3$ from $\hat\var_{\alpha}^2$ by applying the scalar introduced in \eqref{eq:c.star2}. Indeed, observe that 
\[
\hat\var^3(\boldsymbol{X})= -\hat\mu(\boldsymbol{X}) +c^*(\hat\var^2(\boldsymbol{X})+\hat\mu(\boldsymbol{X})),
\]
where
\begin{equation}\label{eq:scalar.gaussian}
\textstyle c^*=\frac{\sqrt{\frac{n+1}{n}}t_{n-1}^{-1}(0.01)}{\Phi^{-1}(0.01)}
\end{equation}
is the optimal scalar. Even if we do not know the theoretical value of the Gaussian scalar $c^*$ given in \eqref{eq:scalar.gaussian}, we can approximate it by solving the optimisation problem stated in \eqref{eq:c.star2}. We refer to Section~\ref{S:specific.scaling}, where exemplary scalar derivation algorithms are presented. For illustration, we present the size of $c^*$ as a function of the underlying sample size $n\in\bN$ as well as confidence threshold $\alpha\in (0,1)$, see Figure~\ref{F:ex1}.
\begin{figure}[htp!]
\begin{center}
\includegraphics[width=0.45\textwidth]{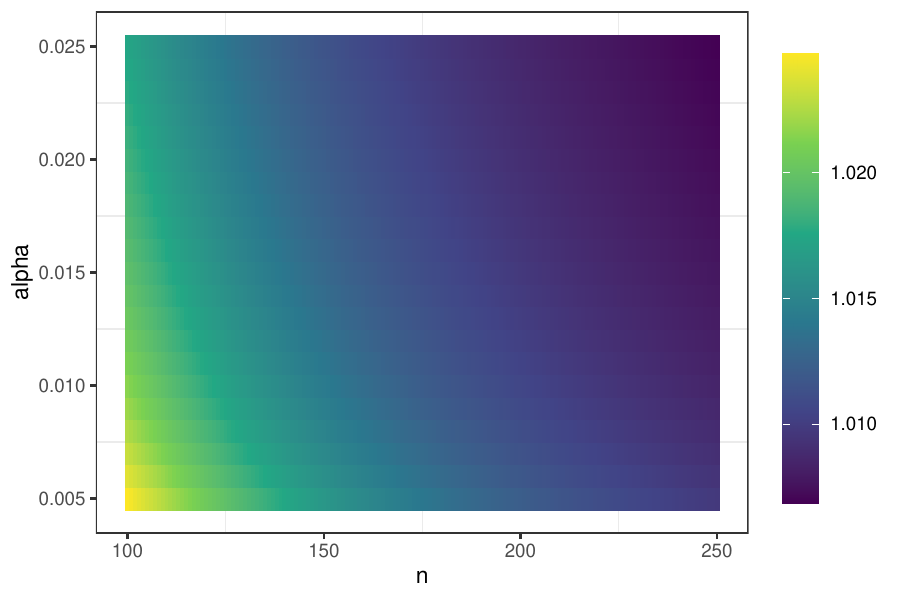}
\end{center}
\vspace{-0.5cm}
\caption{The heatmap presents the values of the Gaussian scalar $c^*$ from \eqref{eq:scalar.gaussian} for different sample sizes $n\in [100,250]$ and confidence thresholds $\alpha\in (0.5\%,2.5\%)$ under the setting described in Example~\ref{ex:1}. Note that the smaller the sample size and confidence threshold, the bigger the scalar.}\label{F:ex1}
\end{figure}

As expected, the smaller the sample size and confidence threshold, the bigger the scalar size. For $n=250$ and $\alpha=1\%$ the size of the adjustment is almost negligible and equal to approximately 1.008. Still, in the low sample setting, it might be material. For example, if we decrease the sample size to $n=50$ and keep $\alpha=1\%$, the scalar increases to approximately 1.044. In the extreme case, setting $n=30$ (annual scenarios) and $\alpha=0.05\%$ (economic capital confidence threshold), the scalar is equal to approximately 1.131.
\end{example}

\begin{example}[Parametric 1-day ES estimator under GPD]\label{ex:2}
The generalized Pareto distribution  (GPD) is often used to model  fat tails, see \cite{McN1999}. Let us fix $\rho=\textrm{ES}_{\alpha}$ for a pre-define confidence level $\alpha\in (0,1)$. Given the threshold parameter $u\in\bR$, shape parameter $\xi<1$, and scale parameter $\beta>0$, the true ES for $X\sim \textrm{GPD}(u,\xi,\beta)$ at confidence level $\alpha\in (0,1)$ is given by
\begin{equation}\label{eq:GPD.es}
\textrm{ES}_{\alpha}(X)=\frac{\var_{\alpha}(X)}{1-\xi}+\frac{\beta+\xi u}{1-\xi},
\end{equation}
where $\var_{\alpha}(X)=-u+\frac{\beta}{\xi}\left(\alpha^{-\xi}-1\right)$. The plug-in GPD estimator is constructed by plugging in the estimated values of $u$, $\xi$, and $\beta$ into \eqref{eq:GPD.es}. It has been shown in \cite{PitSch2022} that this approach often leads to a biased risk estimation. The proposed correction, based on estimated parameters modifications, can be embedded into our framework. It is enough to note that applying a linear change to scale parameter $\beta>0$ is equivalent to scaling. It should be emphasized that the value of the scalar depends on the underlying shape $\xi<1$, so that one should either pre-assume shape value, follow the bootstrap method, or use the robust version of the scalar introduced in~\eqref{eq:c.star3}. For example, let us assume that $X\sim \textrm{GPD}(0,\xi,\beta)$ for unknown scale $\beta>0$ and unknown $\xi\in \Gamma$, where $\Gamma$ is a (compact) subset of $(-\infty,1]$. Let $n=50$ and let us assume we want to estimate the risk of $\textrm{ES}_{\alpha}(X)$ at confidence level 5\%. Let the unscaled ES estimator be given by a plug-in estimator obtained from \eqref{eq:GPD.es}, i.e. let
\[
\hat\rho(\boldsymbol{X})=\frac{\hat\beta(\boldsymbol{X})}{1-\hat\xi(\boldsymbol{X})}\left(\frac{1}{\hat\xi(\boldsymbol{X})}\left(0.05^{-\hat\xi(\boldsymbol{X})}-1\right)+1\right),
\]
where $\hat\xi(\boldsymbol{X})$ and $\hat\beta(\boldsymbol{X})$ are estimators of $\xi$ and $\beta$, respectively, that are obtained using the PWM method. Let $c_{\xi}^*$ denote the scalar value defined in \eqref{eq:c.star} for a specific $\xi\in \Gamma$. Then, noting that the scalar value is invariant with respect to shifts in scale $\beta>0$, the robust version of the scalar could be set to 
\[
c^*=\max_{\xi\in \Gamma}c^*_{\xi}.
\]
Recall that while the robust estimator $c^*\cdot \hat\rho(\boldsymbol{X})$ could be risk biased, the bias sign would be negative, see~\eqref{eq:bias.negative}. To illustrate this, using the Monte Carlo method, we computed the values of $c^*_{\xi}$ for $\Gamma=[-0.5,0.5]$; they are presented in Figure~\ref{F:ex2}. 
\begin{figure}[htp!]
\begin{center}
\includegraphics[width=0.45\textwidth]{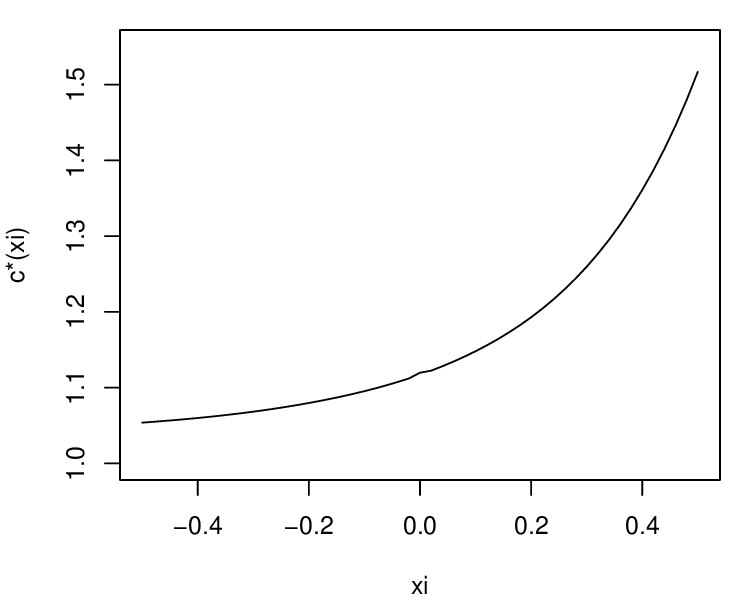}
\end{center}
\vspace{-0.5cm}
\caption{The value of $c^*_{\xi}$ for different $\xi\in [-0.5,0.5]$ under the setting described in Example~\ref{ex:2}. Note that the robust scalar $c^*$ is obtained by taking the maximal value of $c^*_{\xi}$.}\label{F:ex2}
\end{figure}

From the plot, one can deduce that $c^{*}\approx 1.5$ so that the estimator $\rho^*(\boldsymbol{X})=1.5\cdot \hat\rho(\boldsymbol{X})$ does not exhibit positive bias in the GPD setting, for $n=50$ and $\alpha=5\%$. We note that the scalar value sizes presented in Figure~\ref{F:ex2} are consistent with the parameter modification impacts presented in Figure 8 in \cite{PitSch2022}.
\end{example}

\begin{example}[Non-parametric 10-day overlapping VaR estimator under normality]\label{ex:3}
In this example, we show that the standard empirical 10-d VaR estimator based on overlapping P\&Ls is risk biased even if the underlying distribution is Gaussian with vanishing mean, and show how to properly rescale it. Let us fix $\rho=\var_{1\%}$ and assume that $X\sim \cN(0,\sqrt{10}\cdot \sigma)$ for some (unknown) parameter $\sigma>0$. Given a 1-day P\&L (i.i.d.) sample $(\tilde X_i)_{i=1}^{259}$, where $\tilde X_i\sim \cN(0,\sigma)$, we are constructing an overlapping 10-day P\&L sample $\boldsymbol{X}=(X_i)_{i=1}^{n}$ of length $n=250$ by setting $X_i=\tilde X_i+\ldots+\tilde X_{i+9}$, for $i=1,\ldots,n$. Next, assuming (naively) that the sample $\boldsymbol{X}$ is i.i.d.,  we use the standard empirical VaR estimator given in \eqref{eq:varNORM}, that is we set\footnote{Note that the usage of empirical quantile estimators for overlapping  P\&Ls sample is a common approach when estimating Stressed VaR, see Section 5.3 in \cite{EGIM}.}
\begin{equation}\label{eq:ex.3.1}
\hat\rho(\boldsymbol{X})=\hat\var^1(\boldsymbol{X})=-\tfrac{1}{2}(X_{(2)}+X_{(3)}).
\end{equation}
Before we show how to apply the scaling methodology, let us check the risk unbiasedness of $\hat\var^1(\boldsymbol{X})$ using a simple Monte Carlo (MC) simulation. Namely, for $\sigma=1$ and MC size $N=1\,000\,000$ we pick $N$ samples from $(X,\boldsymbol{X})$. Next, we construct a 10-day secured position sample of size $N$ by setting $S=X+\hat\rho(\boldsymbol{X})$,
and use this to approximate the value of \eqref{eq:risk.bias}. From the simulation, we get
\begin{equation}\label{eq:ex3.1}
\var_{1\%}\left(S\right)\approx 0.82,
\end{equation}
which clearly shows the presence of risk bias. To better quantify this effect let us consider the empirical performance measure $\hat T$ defined in \eqref{eq:T.perf2}. Note that in contrast to \eqref{eq:ex3.1}, the value of $\hat T(S)$ will be in fact independent of the (unknown) $\sigma$ parameter and will quantify the confidence threshold $\alpha\in (0,1)$ for which we get $\var_{\alpha}\left(S\right)=0$. Again, using simple Monte Carlo we get
\[
\hat T(S)\approx 1.8\%,
\]
In other words, for $n=250$, the empirical estimator \eqref{eq:risk.bias} based on the 10-day overlapping P\&L sample secures the underlying position for $\var_{1.8\%}$ and not for $\var_{1\%}$.

We are now ready to approximate the size of scalar $c$ that would make the underlying non-parametric estimator unbiased. Namely, using Monte Carlo analysis, we can easily solve optimisation problem \eqref{eq:c.star}. Again, see Section~\ref{S:specific.scaling}, where exemplary scalar derivation algorithms are presented.  In the Normal setting, the numerically obtained scalar value is equal to
\[
c^*\approx 1.14
\]
In other words, one should multiply \eqref{eq:ex.3.1} by 1.14 so that it properly secures $X$ for VaR at level 1\%. If not done, one should expect that approximately 1.8\% exceptions would occur in backtesting. Assuming a similar flaw in the 1-day framework, the regulatory backtest based on annual (daily) time series would on average result in $250*1.8\%=4.5$ exceptions rather than $250\cdot 1\%= 2.5$ exceptions. This would substantially increase the probability of getting into the {\it yellow traffic light} zone, see \cite{Bas1996} for details. In fact, for the unscaled estimator and a perfect i.i.d.normal setting, the probability of getting into the yellow traffic light zone would be equal to 46\%, which seems surprisingly high.\footnote{This number could be easily computed in {\bf R} by using formula~\verb|1-pbinom(4,250,0.018)|.} For completeness, we illustrate this by plotting the values of $\hat T(S(c))$ for different choices of $c\in [1,1.2]$, see Figure~\ref{F:ex3}.

\begin{figure}[htp!]
\begin{center}
\includegraphics[width=0.45\textwidth]{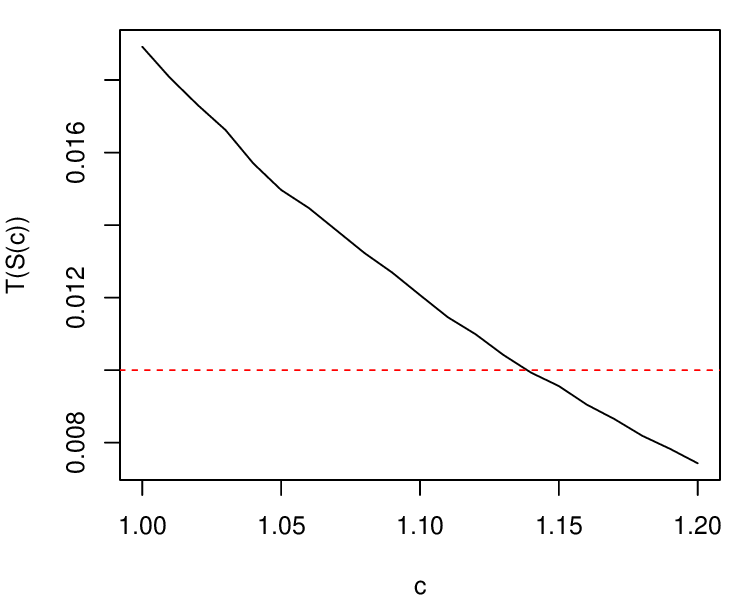}
\end{center}
\vspace{-0.5cm}
\caption{The plot presents the value of empirical performance measure $\hat T(S(c))$ as a function of $c$ under the setting described in Example~\ref{ex:3}. One can see that for $c^*\approx 1.14$ the value $\hat T(S(c^*))$ is close to the desired exception rate $\alpha=1\%$ denoted by the dashed red line. Also, the value of $\hat T(S(c)$ is monotone with respect to $c$, as expected.}\label{F:ex3}
\end{figure}
\end{example}

\begin{example}[Small sample non-parametric 1-day ES estimator under student-$t$ setting]\label{ex:4}

Let us fix $\rho=\textrm{ES}_{2.5\%}$ and assume that $X\sim t_{\nu_0}$ for some (unknown) shape parameter $\nu_0\in [5,\infty)$.\footnote{For simplicity, we do not consider an additional unknown scale parameter. If required, this could be introduced without impacting the results presented in this example.} Let us assume we want to estimate $\textrm{ES}_{2.5\%}(X)$ using a non-parametric estimator for $n=50$. Since $50  \cdot 2.5\%=1.25$, it is expected that only the first-order statistic of our sample will breach the ES-induced conditional threshold (i.e. $\var_{\alpha}(X)$). Consequently, the direct usage of the averaging statistic, as in \eqref{eq:esHS}, might be problematic due to a small sample size. Nevertheless, let us assume that we want to take the average of the three worst-case observations into account and define the unscaled estimator as
\[
\hat\rho(\boldsymbol{X})=-\tfrac{1}{3}(X_{(1)}+X_{(2)}+X_{(3)}).
\]
Then, under additional assumptions imposed on $X$ we can compute the value of a scalar that would make estimator $\hat\rho(\boldsymbol{X})$ unbiased. Since in our setting we assume that $X\sim t_{\nu_0}$ for an unknown parameter $\nu_0\in [5,\infty)$, the robust value of the scalar is equal to
\[
c^*=\sup_{\nu\in [5,\infty)}c_{\nu}^{*},
\]
where $c_{\nu}^{*}$ denotes the scalar value under $\nu\in [5,\infty)$ being the true parameter, see also~\eqref{eq:c.star3}. Using Monte Carlo method we computed the value of $c_{\nu}^{*}$ for different choices of $\nu\in [5,\infty)$, see Figure~\ref{F:ex4}. As expected, the scalar value is monotone with respect to the $\nu$ parameter. From the plot, one can deduce that $c^*\approx 1.55$ would lead to a non-positively biased scaled risk estimator. Note that this example is important also from the modeling perspective as ES depends on the whole left tail, which cannot be adequately captured using order statistics for a finite sample -- our method allows tail risk control with scalar adapted to the assumptions imposed on the tail structure.

\begin{figure}[htp!]
\begin{center}
\includegraphics[width=0.45\textwidth]{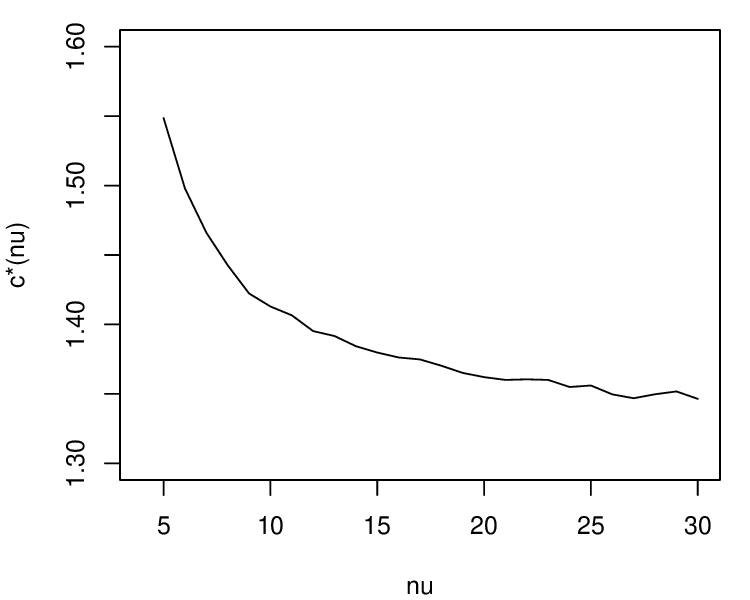}
\end{center}
\vspace{-0.5cm}
\caption{The value of $c^*_{\nu}$ for different $\nu\in [5,\infty)$ under the setting described in Example~\ref{ex:4}. Note that the robust scalar $c^*$ is obtained by taking the maximal value of $c^*_{\nu}$.}\label{F:ex4}
\end{figure}

\end{example}

\subsection{Time and confidence level scaling}\label{S:specific.scaling}
In this section, we focus on more practical situations where the scaling is applied. This refers to both time and confidence threshold scaling in various settings linked to the market risk model as well as its extensions, i.e. economic capital model and exotic risk factor model.

For consistency, in all examples, we pick a common set of distributions $\cF$ for robust scalar value evaluation. The family $\cF$ is based mainly on the $t$-Student family, see Table~\ref{T:ex5} for the exact list of considered distributions. We did this to better illustrate the intricacies built into various frameworks. We also use a common Monte Carlo (MC) size equal to $N=1\,000\,000$. For transparency, we split the scalar estimation procedure into two steps linked to confidence level and time horizon adjustments. Namely, in the first step, we assume that the P\&L sample holding period is the same as the underlying risk holding period, while in the second step, we further adjust the holding period. In this section, we provide a detailed scaling algorithm description for each example together with a detailed comment on the two-step procedure. Also, for brevity, given distribution $\bF$ and $k\in\bN$ we use $\textrm{Conv}(\bF,k)$ to denote $k$ times convolution of $\bF$. Note that a sample from $\textrm{Conv}(\bF,k)$ could be obtained by summing $k$ independent picks from $\bF$ or by taking a direct sample from the convoluted distribution. 

\begin{example}[1-day to 10-day VaR scaling]\label{ex:5}
In this example, we focus on 1-day to 10-day scaling for $\var_{1\%}$. We assume that we are given a 1-day P\&L sample $\boldsymbol{X}$ of size $n=250$ from some unknown zero-mean distribution and use it to estimate 1-day $\var_{1\%}$. For estimation, we use the non-parametric estimator defined in \eqref{eq:varHS}, i.e we set
\begin{equation}\label{eq:varHS2}
\hat\rho(\boldsymbol{X})=-\tfrac{1}{2}(X_{(2)}+X_{(3)}).
\end{equation}
Our goal is to use $\hat\rho(\boldsymbol{X})$ combined with an optimal scalar $c^*>0$ in such a way that the scaled estimator
\[
\hat\rho^*(\boldsymbol{X})=c^*\cdot \hat\rho(\boldsymbol{X})
\]
could be used to estimate 10-day $\var_{1\%}$. In other words, we want to find scalar $c$ that transfers 1-day $\var_{1\%}$ into 10-day $\var_{1\%}$. Please recall that time scaling is a common practice adopted by most financial institutions, see \cite{EBA2022}.

For simplicity, let us assume that the  distribution of the 10-day P\&L, is a convolution of ten independent 1-day P\&Ls.
Assume that the true (unknown) distribution of 1-day P\&L is from a family $\cF$. Then, for any specific choice of $\bF\in \cF$ we can approximate the value of $c^*$. For this purpose, we use  a simple Monte Carlo based approach that is outlined in Figure~\ref{F:code1}.
\begin{figure}[htp!]
\begin{tcolorbox}
{\footnotesize
{\bf Algorithm 1 (1-day to 10-day time shift scalar for VaR).}
}
{\footnotesize
\noindent Fix strong Monte Carlo size $M\in\bN$, distribution $\bF\in\cF$, VaR confidence threshold $\alpha\in (0,1)$, risk estimator $\hat\rho(\cdot)$ and sample size $n\in\bN$.
\begin{enumerate}[(1)]
\item Simulate $M$ times 1-day P\&L sample from $\bF^n$. Denote the $m$th realisation by $\boldsymbol{x}^m:=(x^m_1,\ldots,x^m_n)$.

\item Simulate $M$ times 10-day P\&L from $\textrm{Conv}(\bF,10)$. Denote the $m$th realisation as $\tilde x^m$.

\item For any $c>0$ construct the  {\it secured position sample} $\boldsymbol{S}(c):=(S^1(c),\ldots,S^M(c))$, where 
$
S^m(c):= \tilde x^m+c\hat\rho(\boldsymbol{x}^m).
$
\item Approximate $c^*$ by solving empirical equivalent of \eqref{eq:rho.star}. For example, this could be done by setting
\[
\textstyle c^*:=\argmin_{c>0}\left|S^{(\lfloor M\alpha \rfloor)}(c)\right|,
\]
where $S^{(k)}(c)$ is the $k$th order statistic of sample $\boldsymbol{S}(c)$.\footnote{this is a non-parametric estimator of $\var_{\alpha}$ for $\boldsymbol{S}(c)$.}
\end{enumerate}
}
\end{tcolorbox}
\caption{The algorithm for determining 1-d to 10-d time shift scalar for VaR.} 
\label{F:code1}
\end{figure}

Since the methodology presented in Figure~\ref{F:code1} takes care of both time and confidence level scaling, we decided to also present the results for an adjusted algorithm in which we first rescale \eqref{eq:varHS2} to be unbiased for 1-day and then apply algorithm from Figure~\ref{F:code1} to the rescaled version of \eqref{eq:varHS2}. Consequently, we obtain two scalars which could be multiplied to obtain the final scalar. For brevity, we refer to the first scalar as {\it confidence scalar} and to the second as {\it time scalar}. Note that the confidence scalar is easily derived using a slight modification of the algorithm presented in Figure~\ref{F:code1}. Namely, we need to simply replace 10-day P\&L with 1-day P\&L in step (2).

The results for various distributional choices are presented in Table~\ref{T:ex5}.
\begin{table}[htp!]
\centering
{\small
\begin{tabular}{lr|rr}
  \hline
 Distribution &  $c^*$ & Conf. $c^*$ & Time $c^*$ \\ 
  \hline 
  Laplace & 2.74 & 0.98 & 2.78 \\ 
  student-$t$ ($\nu=3$) & 2.99 & 0.98 & 3.06 \\ 
  student-$t$ ($\nu=5$) & 2.90 & 0.99 & 2.93 \\ 
  student-$t$ ($\nu=7$) & 2.94 & 0.99 & 2.98 \\ 
  student-$t$ ($\nu=10$) & 2.99 & 0.99 & 3.01 \\ 
  student-$t$ ($\nu=20$) & 3.06 & 0.99 & 3.09 \\ 
  student-$t$ ($\nu=30$) & 3.09 & 0.99 & 3.12 \\ 
  Normal & 3.14 & 0.99 & {\bf 3.16} \\ 
  GN(3) & 3.41 & 1.00 & 3.43 \\ 
  Cauchy & 9.17 & 0.93 & 9.91 \\ 
   \hline
\end{tabular}
}
\caption{The table presents the value of 1-d to 10-d scalars for empirical VaR estimator and various distributions under the setting described in Example~\ref{ex:5}; GN(3) stands for generalised normal distribution with shape parameter $\beta=3$. One can see that the adjusted scalar value could be materially different from $\sqrt{10}\approx 3.16$ even in a simple i.i.d. setting, depending on the underlying distribution assumption.
}\label{T:ex5}
\end{table}
As expected, for the Normal distribution the value of the time scalar is close to $\sqrt{10}\approx 3.16$ as for this case the square-root-of-time approach is a viable option,  see Section~\ref{S:square.root}. On the other hand, we see that for other distributions the time scalar could substantially differ from $\sqrt{10}$. In the extreme case of Cauchy distribution, the scalar value is equal to 10; note that our algorithm was able to correctly approximate the value of 10 which is due to the fact that $\textrm{Conv}(X,10) \sim 10\cdot X$ for independent Cauchy distributed random variables. 

It should be also noted that for student-$t$ distributions, which are considered more heavy-tailed than Gaussian, the scalar value is in fact smaller than $\sqrt{10}$ so that the application of square root of time might lead to over-conservative capital reserves. This phenomenon could be easily explained by the central limit theorem: if the underlying distribution has a finite variance, then the sum of i.i.d. random variables should tend to normal distribution -- while the variance of the sum would scale properly, the sum tail shape will change and be closer to normal. In consequence, the naive square-root-of-time scalar usage would result in an over-conservative value, as the scalar should be reduced by the ratio of the corresponding quantiles. For example, for the student-$t$ distribution with parameter $v=5$ and $\alpha=1\%$ we have
\begin{equation}\label{eq:square.adjust}
\textstyle \sqrt{10} \frac{\Phi^{-1}(\alpha)}{t^{-1}_{v}(\alpha)}\approx 2.82
\end{equation}
which is relatively close to the value 2.93 obtained in Table~\ref{T:ex5}; the values are not equal to each other, as the quantile ratio in \eqref{eq:square.adjust} is based on the limit CLT argument. This shows that our method can automatically handle certain subtleties built into the estimation processes.

Finally, we note that in this example, the confidence scaling is almost negligible due to a relatively large sample size $n=250$ and moderate confidence threshold value $\alpha=1\%$.

\end{example}

\begin{example}[Exotic risk VaR scaling]\label{ex:6}
In this example, we focus on a situation, when the size of the statistical sample is low when compared to the underlying risk confidence threshold and we get the holding period mismatch induced by limited data availability. This is typically linked to positions depending on exotic risk factors whose values cannot be directly inferred from the available market data or no full stress period data is available. Those types of risks are typically modeled via risk add-on frameworks, such as {\it Risks not in VaR} (RNIV), {\it Risks not in the model engines} (RNIME), or {\it Non-modellable risk factors} (NMRF), see~\cite{EGIM}, \cite{PRA2018}, and \cite{EBA.NMRF} for regulatory backgrounds. In this situation, we want to provide a capital reserve whose value is consistent with the base risk metric having only a limited dataset at hand. For simplicity, we focus on the VaR case; similar logic could be used in the case of ES.

Let us assume that our base metric is 10-day $\var_{1\%}$ but we are given a statistical sample  consisting of twelve 1-month P\&L observations denoted as $\boldsymbol{X}=(X_1,\ldots,X_{12})$. For simplicity, let us assume that we want to use the worst-case outcome for capital estimation, i.e. we set the non-scaled 10-day $\var_{1\%}$ estimator to
\[
\hat\rho(\boldsymbol{X})=-X_{(1)}.
\]
As usual, we look for a constant $c^*$ that makes $\hat\rho$ risk unbiased. Note that the constant needs to account for time horizon downscaling (from 1-month to 10-day) as well as risk confidence threshold upscaling as 12 observations are not sufficient to meaningfully estimate tail 1\% quantile in a non-parametric way. To do this, we have to impose some additional assumptions on our framework. 

As in the previous example, let us assume that the 10-day P\&Ls distribution belongs to class $\cF$. Also, let us assume that the given 1-month P\&Ls are i.i.d. and are a sum of two independent 10-day P\&Ls. In that case, we can estimate $c^*$ using the robust approach from~\eqref{eq:c.star3} similarly as in Example~\ref{ex:4}. Using the algorithm presented in Figure~\ref{F:code2} we can compute the scalar for any $\bF\in\cF$; note that the scalar is invariant to variance changes, so it is sufficient to consider unit variance distributions. 

\begin{figure}[htp!]
\begin{tcolorbox}
{\footnotesize
{\bf Algorithm 2 (VaR scaling for monthly data).}
}
{\footnotesize
\noindent Fix strong Monte Carlo size $M\in\bN$, distribution $\bF\in\cF$, VaR confidence threshold $\alpha\in (0,1)$, risk estimator $\hat\rho(\cdot)$ and sample size $n\in\bN$.
\begin{enumerate}[(1)]
\item Simulate $M$ times 20-day P\&L sample from $\boldsymbol{X}\sim (\bF*\bF)^n$. Denote the $m$th realisation by $\boldsymbol{x}^m:=(x^m_1,\ldots,x^m_n)$.

\item Simulate $M$ times 10-day P\&L sample from $\boldsymbol{X}\sim \bF$. Denote the $m$th realisation (sum) as $\tilde x^m$.

\item For any $c>0$ construct the  {\it secured position sample} $\boldsymbol{S}(c):=(S^1(c),\ldots,S^M(c))$, where 
$
S^m(c):= \tilde x^m+c\hat\rho(\boldsymbol{x}^m).
$
\item Approximate $c^*$ by solving empirical equivalent of \eqref{eq:rho.star}. For example, this could be done by setting
\[
\textstyle c^*:=\argmin_{c>0}\left|S^{(\lfloor M\alpha \rfloor)}(c)\right|,
\]
where $S^{(k)}(c)$ is the $k$th order statistic of sample $\boldsymbol{S}(c)$.
\end{enumerate}
}
\end{tcolorbox}
\caption{The algorithm for determining a VaR scalar when only limited monthly data is available. The robust version of the scalar could be obtained by taking maximum over all distributions in $\cF$.} 
\label{F:code2}
\end{figure}

As before, we split the scalar estimation process into two steps. The confidence scalar is easily derived using a slight modification of the algorithm presented in Figure~\ref{F:code2}. Namely, we need to use monthly P\&L in step (2). The resulting values are presented in Table~\ref{T:ex6}. 
\begin{table}[htp!]
\centering
{\small
\begin{tabular}{lr|rr}
  \hline
Distribution & c* & Conf. c* & Time c* \\ 
  \hline
  Laplace & 1.80 & 2.49 & 0.72 \\ 
  student-$t$ (nu=3) & 1.94 & 2.77 & 0.71 \\ 
  student-$t$ (nu=5) & 1.71 & 2.39 & 0.71 \\ 
  student-$t$ (nu=7) & 1.64 & 2.29 & 0.71 \\ 
  student-$t$ (nu=10) & 1.59 & 2.22 & 0.71 \\ 
  student-$t$ (nu=20) & 1.53 & 2.16 & 0.71 \\ 
  student-$t$ (nu=30) & 1.52 & 2.13 & 0.71 \\ 
  student-$t$ (nu=50) & 1.51 & 2.13 & 0.71 \\ 
  student-$t$ (nu=100) & 1.50 & 2.12 & 0.71 \\ 
  Normal & 1.49 & 2.10 & 0.71 \\ 
  GN(3) & 1.40 & 2.02 & 0.69 \\ 
  Cauchy & 4.51 & 9.04 & 0.50 \\ 
   \hline
\end{tabular}
}
\caption{The table presents the value of scalars for exotic risk factor empirical estimator and various distributions under the setting described in Example~\ref{ex:6}; GN(3) stands for generalised normal distribution with shape parameter $\beta=3$. One can see that while the time scalar is rather stable, the confidence scalar depends strongly on the underlying distribution.}\label{T:ex6}
\end{table}
\end{example}

As expected, in this case, the impact of the confidence scalar on the final scalar is much more profound due to the relatively small sample size. Note that while time scalars are in fact smaller than one as we down-scale the risk, the confidence scalars are high due to the relatively small size of the sample. Also, it should be noted that while time scaling looks robust, the size of the confidence scalar is sensitive to the choice of the underlying distribution.

\begin{example}[Economic capital ES risk scaling]\label{ex:7}
In the last example, we show how the method introduced in this paper could be used to help estimate economic capital risk within ES framework. Namely, given the limited P\&L sample, we want to scale both the confidence threshold and time horizon to a more extreme setting. 

While the standard market risk is estimated for shorter holding periods, e.g. 10-day, and relatively non-extreme tail thresholds, e.g. using $\var_{1\%}$ or $\textrm{ES}_{2.5\%}$, the economic capital is typically based on longer holding periods, e.g. 1 year, and extreme risk thresholds, e.g. 0.1\%; see Chart 30 and Chart 31 in \cite{ICAAP.pr} for more details on banks' ICAAP practices.

For simplicity, assume that the economic capital holding period is equal to 250 days and the underlying economic capital risk measure is $\textrm{ES}_{0.1\%}$. Moreover, let us assume that we want to utilise the Pillar 1 framework in which 10-day non-overlapping P\&Ls could be produced for the last 30 years of data giving us a total of $n=750$ historical observations denoted by $\boldsymbol{X}=(X_1,\ldots,X_n)$. We also assume that we are given an (unscaled) risk estimator defined as
\[
\hat \rho(\boldsymbol{X})=-\frac{1}{6}\sum_{i=1}^{6}X_{(i)},
\]
which could be seen as an empirical estimator of 10-day $\textrm{ES}_{0.8\%}$ as $750*0.8\%=6$. We decided to use a higher (initial) confidence threshold for the unscaled estimator as for target 0.1\% thresholds in the non-parametric setting we have $750\cdot 0.1\% =0.75$, and the usage of a single worst-case observation might be non-robust. The objective of this example is to show how to scale $\hat\rho$ so that it represents economic capital risk for the annual holding period and confidence threshold 0.1\%.
 We follow a standard i.i.d. setting and use a similar scheme as in Example~\ref{ex:5}.

An exemplary estimation algorithm for a given distribution $\bF\in\cF$ is presented in Figure~\ref{F:code3}. The confidence scalar is obtained by modifying the estimation algorithm in step (1), in which 10-day P\&L is replaced by 250-day P\&L. 
\begin{figure}[htp!]
\begin{tcolorbox}
{\footnotesize
{\bf Algorithm 3 (Economic Capital risk scaling for ES).}
}
{\footnotesize
\noindent Fix (big) strong Monte Carlo size $M\in\bN$, distribution $\bF\in\cF$, ES confidence threshold $\alpha\in (0,1)$, risk estimator $\hat\rho(\cdot)$ and sample size $n\in\bN$.
\begin{enumerate}[(1)]
\item Simulate $M$ times 10-day P\&L sample from $\boldsymbol{X}\sim \bF^n$. Denote the $m$th realisation by $\boldsymbol{x}^m:=(x^m_1,\ldots,x^m_n)$.

\item Simulate $M$ times 250-day P\&L from $\textrm{Conv}(X,25)$. Denote the $m$th realisation as $\tilde x^m$.

\item For any $c>0$ construct the  {\it secured position sample} $\boldsymbol{S}(c):=(S^1(c),\ldots,S^M(c))$, where 
$
S^m(c):= \tilde x^m+c\hat\rho(\boldsymbol{x}^m).
$
\item Approximate $c^*$ by solving empirical equivalent of \eqref{eq:rho.star}. For example, this could be done by setting
\[
\textstyle c^*:=\argmin_{c>0}\left|\tfrac{1}{\lfloor M\alpha\rfloor}\sum_{i=1}^{\lfloor \alpha M \rfloor}S^{(i)}(c)\right|,
\]
where $S^{(k)}(c)$ is the $k$th order statistic of sample $\boldsymbol{S}(c)$.\footnote{ this is a non-parametric estimator of $\textrm{ES}_{\alpha}$ for $\boldsymbol{S}(c)$.
}
\end{enumerate}
}
\end{tcolorbox}
\caption{The algorithm for determining an ES scalar in the economic capital setting. The robust version of the scalar could be obtained by taking maximum over all distributions in $\cF$.} 
\label{F:code3}
\end{figure}
The results are presented in Table~\ref{T:ex7}. In this case both confidence and time scalar depend on the underlying distribution. That said, it should be noted that the inverse relation could be observed, i.e. the bigger the confidence scalar the smaller the time scalar.

\begin{table}[htp!]
\centering
{\small
\begin{tabular}{lr|rr}
  \hline
Distribution & c* & Conf. c* & Time c* \\ 
  \hline
 Laplace & 5.16 & 1.47 & 3.51 \\ 
  student-$t$ (nu=3) & 7.78 & 2.27 & 3.43 \\ 
  student-$t$ (nu=5) & 5.62 & 1.67 & 3.35 \\ 
  student-$t$ (nu=7) & 5.64 & 1.55 & 3.65 \\ 
  student-$t$ (nu=10) & 5.87 & 1.47 & 3.99 \\ 
  student-$t$ (nu=20) & 6.07 & 1.35 & 4.50 \\ 
  student-$t$ (nu=30) & 6.19 & 1.32 & 4.70 \\ 
  student-$t$ (nu=50) & 6.23 & 1.30 & 4.78 \\ 
  student-$t$ (nu=100) & 6.34 & 1.29 & 4.93 \\ 
  Normal & 6.26 & 1.27 & 4.95 \\ 
  GN(3) & 7.09 & 1.20 & 5.92 \\ 
   \hline
\end{tabular}
}
\caption{The table presents the value of scalars for economic capital estimator under various distributions under the setting described in Example~\ref{ex:7}; GN(3) stands for generalised normal distribution with shape parameter $\beta=3$. One can see that both confidence and time scalars depend on the underlying distribution.}\label{T:ex7}
\end{table}
\end{example}

\section{Empirical analysis}\label{S:numerical}
In this short section, we present a simple empirical study that illustrates the impact of the scaling method choice on capital reserve adequacy. We do this to illustrate how the choice of the underlying scaling methodology could impact capital reserve conservativeness and backtesting performance.  We want to emphasize that the analyses in this section provide only a top-level illustration and are based on simple benchmark scaling methodologies. Also, instead of developing an omnibus scaling method (which in practice is almost impossible), we want to show that the inclusion of risk unbiasedness concept in the scalar calibrating could lead to model improvement. Rather than being a simple alternative scaling methodology, our method constitutes a statistical framework that can be used to refine or verify any pre-existing plug-in scaling methodology; see Example~\ref{ex:2} for application to extreme value theory methods.

In this section, we consider various parametric scaling methods outlined in Section~\ref{S:scaling} and Section~\ref{S:examples} as well as two exemplary data-driven scalar estimation methodologies that are based on Example~\ref{ex:6} setup.

For simplicity, we perform the calculations using return rates of multiple (equity) portfolios. We use data from {\it $q$-factors data library}, see \cite{HouXueZha2022,HouXueZha2023}. Namely, we take weekly value-weighted returns of testing portfolios based on 188 anomalies with 1-way sorts. In total, we consider 1853 different portfolios and perform a backtesting exercise on data from 2010 to 2021, i.e. we take the last $N:=625$ observations for each portfolio, see Figure~\ref{F:empirical1} for exemplary portfolio data illustration.

\begin{figure}[htp!]
\begin{center}
\includegraphics[width=0.5\textwidth]{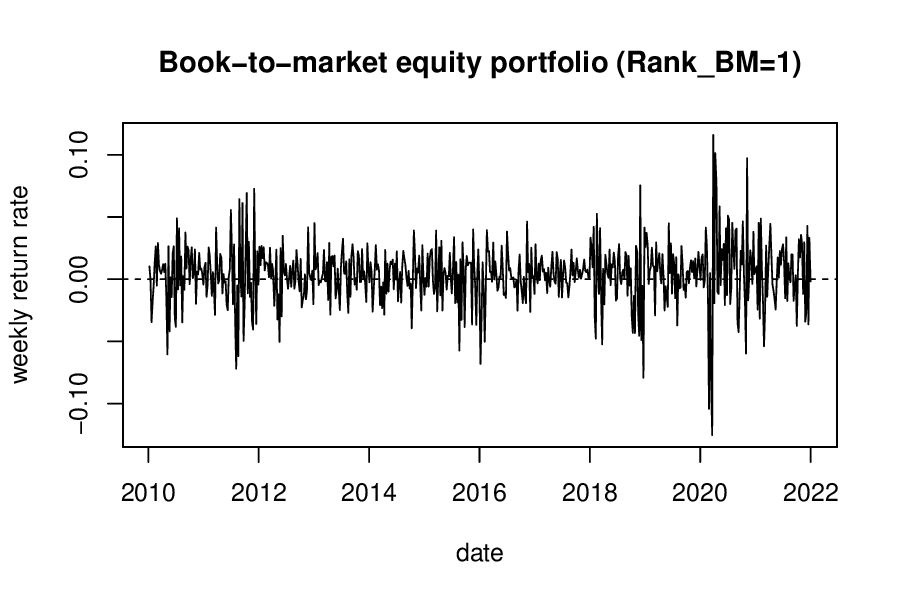}
\end{center}
\vspace{-0.5cm}
\caption{The plot presents 2010-2021 weekly return rates for an exemplary single {\it testing portfolio} based on Book-to-market equity anomaly for the first decile sorting. In total, 1853 portfolios are considered.}\label{F:empirical1}
\end{figure}

For the analysis, we set $\var_{1\%}$ as the underlying reference risk measure and consider two holding periods: one week (5-day) and two weeks (10-day). In both cases, we estimate the risk using the annual time series of non-overlapping weekly returns, i.e. we set the estimation sample size to $n:=50$. 
Then, we estimate risk reserves by combining first-order statistics with various scaling methods and check the output performance by counting the exception rates for the standard rolling windows backtest.

Let us now describe the testing framework in more detail; for simplicity, we focus on the 5-day holding period case. Given a single portfolio weekly returns $\boldsymbol{R}:=(R_1,\ldots,R_{N})$, we set the $t$th backtesting day {\it unscaled VaR risk estimator} as the $1$st order statistic, that is, we set 
\[
\hat\rho_t := -\min(R_t,\ldots,R_{n+(t-1)}), \quad t\in\{1,N-n\}.
\]
The corresponding $t$th backtesting day 5-day {\it realised return} is given by
\begin{equation}\label{eq:5day}
X_t := R_{n+t}, \quad t\in\{1,N-n\}.
\end{equation}
Now, given scalar value $c\in \bR$ we calculate the $t$th backtesting day {\it scaled secured position} as
\[
S_t(c):=X_t +c\cdot\hat\rho_t,\quad t\in\{1,N-n\}.
\]
Finally, as in Section~\ref{S:backtesting}, we consider the exception rate statistic given by
\[
\hat{T}(c):=\sum_{t=1}^{N-n}\frac{\1_{\{S_t(c)\leq 0\}}}{N-n}
\]
and use it as a key performance measure that evaluates estimated capital reserve conservativeness.
We repeat this procedure for each portfolio in scope and various risk scaling methodologies. We also follow the same framework for the 10-day holding period setting. In the case of 10-day holding period, \eqref{eq:5day} is replaced by $X_t=R_{t+n}+R_{t+n+1}$, i.e. we linearly aggregate 5-days returns to obtain (overlapping) 10-day return. The exact list of scalar derivation methods considered in this section together with scalar values is presented in Table~\ref{T:emp1}.

\begin{table}[htp!]
\centering
{\small
\begin{tabular}{ll|rr}
  \hline
\# &Scaling method & 5-day  & 10-day \\ 
  \hline
1&  Non-Scaled + SQRT rule & 1.00 & 1.41 \\ 
2& Normal Ratio + SQRT rule & 1.13 & 1.60 \\ 
 3& Normal unbiased & 1.15 & 1.62 \\ 
 4&  student-$t$ unbiased ($\nu=6$)  & 1.23 & 1.70 \\
  \hline
 5&  Empirical unbiased (student-$t$)  & (1.25) & (1.73)\\
 6&  Empirical unbiased (Non-par) & (1.29) & (1.74)\\ 
   \hline
\end{tabular}
}
\caption{The table presents the summary of the scaling methods considered in the empirical analysis. While fixed scalar values are used in the first four methods, the last two methods are based on data-driven scalars that are fit on portfolio-level -- mean values are presented.}\label{T:emp1}
\end{table}
The first method (\#1) is introduced to show the baseline method when no scaling is applied in the 5-holding period setting; for better comparability, we decided to apply square-root-of-time scalar of size $\sqrt{2}\approx 1.41$ when shifting the risk from 5-day holding period to 10-day holding period. 

The second method (\#2) assumes that the first-order statistic could be used to (empirically) approximate any quantile smaller than $1/50 =2\%$ and is based on the quantile scaling approach. Namely, we use Gaussian distribution based conservative scalar of size
\[
d_1(0.01,0.02)=\frac{\Phi^{-1}(0.01)}{\Phi^{-1}(0.02)}\approx 1.13.
\]
for the 10-day holding, we rescale 5-day estimator using the square-root-of-time approach. 

The third (\#3) and fourth (\#4) methods are based on the parametric risk unbiased scaling method introduced in Example~\ref{ex:6}.  The scalar values are derived using a parametric approach in which we assume that the underlying distribution is either Normal (\#3) or student-$t$ with shape parameter $\nu=6$ (\#4). 

The last two methods (\#5 and \#6) are data-driven. In the first case (\#5), for each portfolio, we estimate the number of degrees of freedom based on past data (up to 2010) and then use the parametric method from Example~\ref{ex:6}. The value in Table~\ref{T:emp1} corresponds to the mean scalar value from all portfolios in scope. In the second case (\#6), we use a non-parametric approach in which we calculate exception rate statistic on past data using a rolling window backtest and fit the scalar in such a way, that this statistic is equal to the target (1\%) level. Again, Table~\ref{T:emp1} value corresponds to the mean scalar value for all portfolios in the scope of the analysis.

It should be noted that while the last two methods (\#5 and \#6) are data-driven, they are not adaptive, i.e. we fit the empirical scalars using past data (pre-2010) and do not adjust the scalars during the backtest, e.g. using data available up to a specific backtest day. We do this so that those estimators can be more directly compared to the previous methods. Note that in a production environment, the scalar could be re-calibrated periodically to increase estimation efficiency. Since the method is not resource-consuming, the re-calibration could be made even daily.

The results of the backtest are presented in Table~\ref{T:emp2}. 
\begin{table}[htp!]
\centering
{\small
\begin{tabular}{l|rr|rr|l}
 \multirow{2}{1em}{\#} &  \multicolumn{2}{c|}{5D} & \multicolumn{2}{c|}{10D} & \multirow{2}{1em}{Best}\\
 & mean & sd  & mean & sd &\\ 
  \hline
  1  & 2.21 & 0.32 & 1.93 & 0.40 & \phantom{0}1\%  \\ 
  2  & 1.59 & 0.24 & 1.49 & 0.34 &17\%\\ 
  3  & 1.52 & 0.24 & 1.45 & 0.33 &\phantom{0}6\%\\ 
  4  & 1.24 & 0.22 & 1.33 & 0.30 &{\bf 38\%}\\ 
  5  & 1.19 & 0.22 & 1.29 & 0.27 &15\%\\ 
  6  & {\bf 1.07} & 0.25 & \bf{1.28} & 0.27 &23\%\\ 
\end{tabular}
}
\caption{The table summarises the results of the empirical backtesting exercise for 1853 different portfolios and 6 scaling methods from Table~\ref{T:emp1}. Results for both 5-day and 10-day holding periods are presented. Column {\it mean} presents the mean of exception rates and column {\it sd} presents the corresponding standard deviation from the mean. Column {\it Best} indicates the aggregated percentage of portfolios for which the scaling method's exception rate was closest to the target 1\% value in absolute terms.}\label{T:emp2}
\end{table}
 We see that empirical scaling method \#6 given exception rates which are (on average) closest to the reference value $\alpha=1\%$. On the other hand, the classical methods based on the square-root-of-time rule and normal quantile ratios (\#2 and \#3) lead to relatively large exception rates. One can also note that for method \#4, the resulting exception rate is most frequently closest to 1\% by looking at {\it Best} column. 
 To better understand this phenomenon and the interaction between method \#4 and \#6, we decided to plot empirical exception rate densities, see Figure~\ref{F:empirical2} and Figure~\ref{F:empirical3}.  For each method, the empirical density is constructed by taking 1853 exception rates and smoothing the corresponding histogram.

 \begin{figure}[htp!]
\begin{center}
\includegraphics[width=0.4\textwidth]{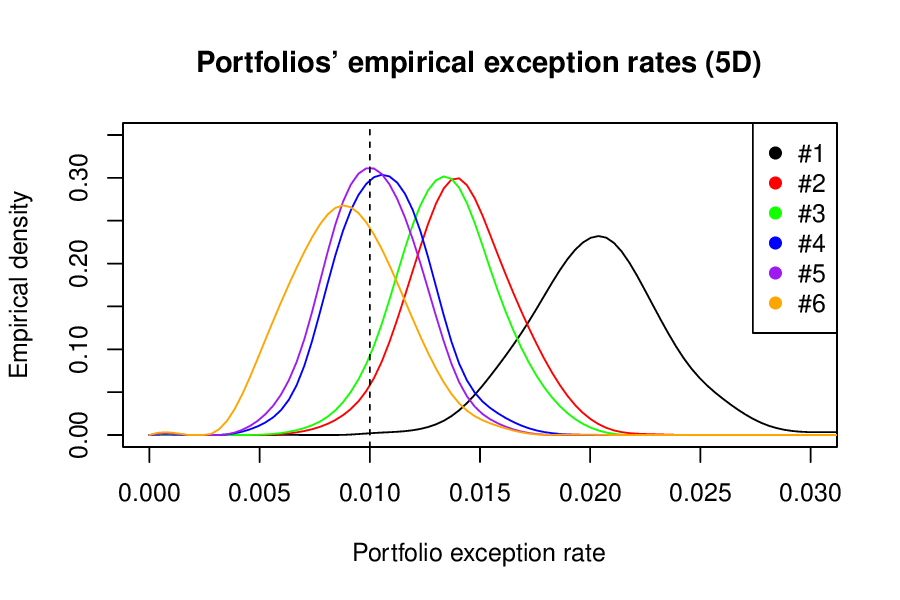}
\end{center}
\vspace{-0.5cm}
\caption{The plot presents 5-day holding period empirical exception rate densities for all portfolios. We see that while exception rates for methods \#1-\#3 are typically bigger than the target $\alpha=1\%$ exception rates, methods \#4-\#6 give results closer to the desired level.}
\label{F:empirical2}
\end{figure}

\begin{figure}[htp!]
\begin{center}
\includegraphics[width=0.4\textwidth]{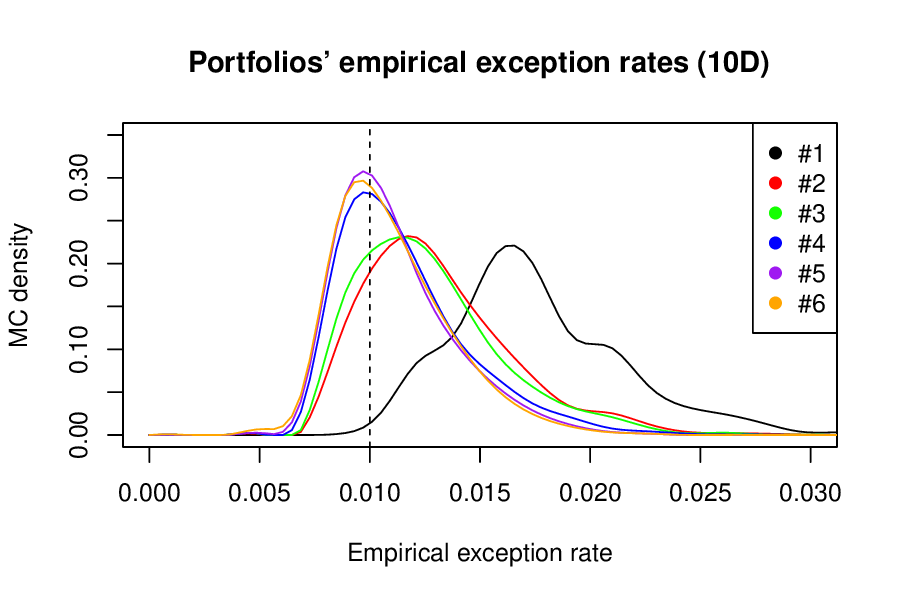}
\end{center}
\vspace{-0.5cm}
\caption{The plot presents 10-day holding period empirical exception rate densities for all portfolios. While the top-level conclusions remain the same as in Figure~\ref{F:empirical2}, we note that exception rates for methods \#1-\#3  are slightly closer to 1\% due to the square-root-of-time rule conservativeness.}\label{F:empirical3}
\end{figure}

 From the plots, we observe that methods \#1-\#3 tend to produce exception rates that are bigger than the target $\alpha=1\%$ confidence threshold. On the other hand, methods \#4-\#6 produce histograms that are more centered around 1\% value. Closer looks into 5-day holding period data explains why the mean exception rate for method \#1 is closer to 1\% but does not give the closest exception rates for all portfolios, cf. column {\it Best} in Table~\ref{T:emp2}. Namely, this is due to the non-parametric nature of method \#1 which produces more conservative scalar values which in turn break the symmetry in the right tail of the empirical exception rate density. It should be also noted that the holding period change from 5-days to 10-days impacts the shape of the empirical densities. In particular, the left tails seem to be heavier, which results in higher exception rates in almost all considered cases, see Table~\ref{T:emp2}.

Finally, to check all methods' stability we also re-run the exercise on two simulated data sets. In the first case, we replaced all portfolio data points with i.i.d. normal observations, while in the second case, we use i.i.d. student-$t$ observations with $\nu=6$. See Figure~\ref{F:empirical4} for the obtained empirical densities. As expected, in the first case, methods based on normal assumption, i.e. \#2 and \#3 had the best results and were closely followed by non-parametric method \#6. The mean exception rate for method \#2, \#3, and \#6, was equal to $1.11\%$, $1.03\%$, and $1.05\%$, respectively. In the second case, methods \#4, \#5, and \#6 had the best performance, while methods \#2 and \#3 led to slight risk underestimation. The mean exception rate for \#4, \#5, and \#6 was equal to 0.99\%, 1.00\%, and 1.02\%, respectively.

\begin{figure}[htp!]
\begin{center}
\includegraphics[width=0.4\textwidth]{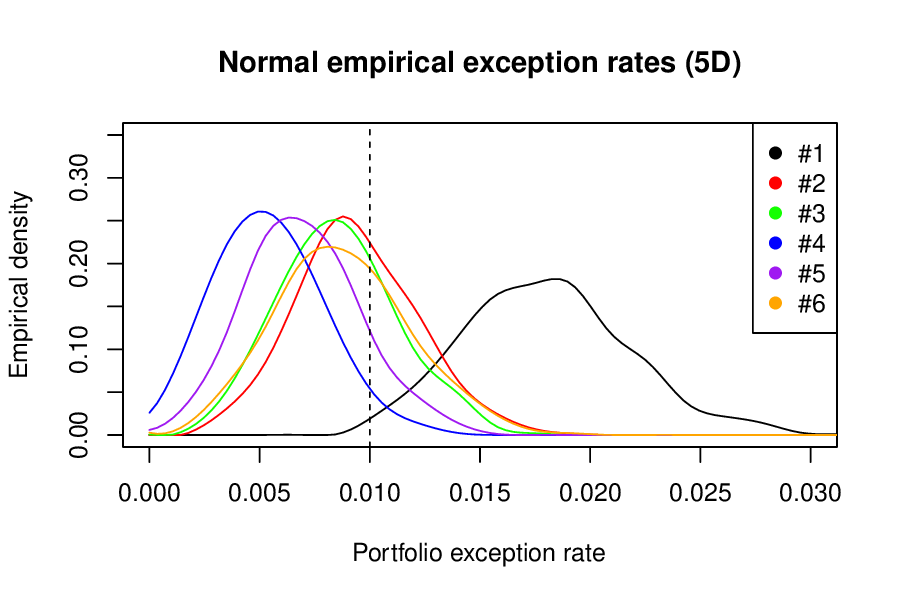}\\
\includegraphics[width=0.4\textwidth]{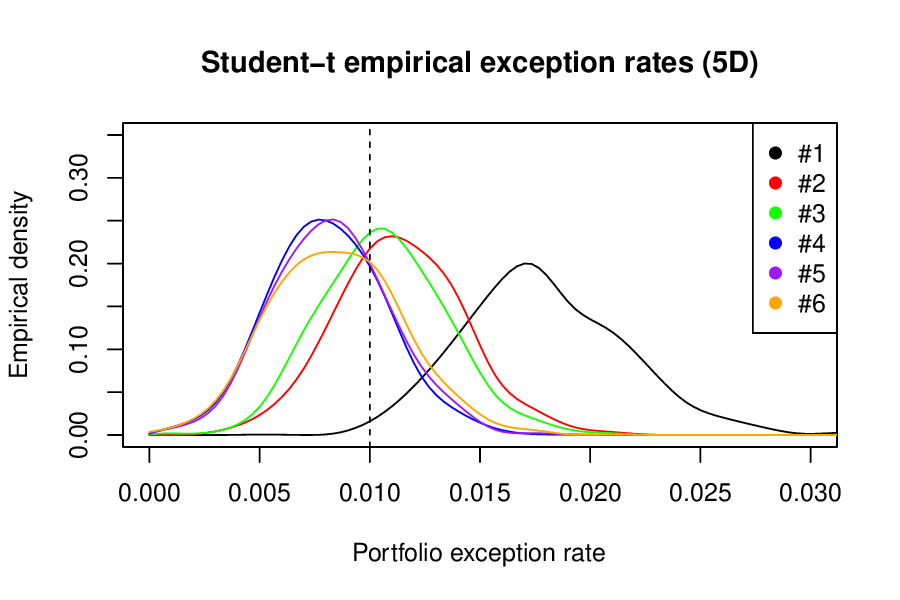}
\end{center}
\vspace{-0.5cm}
\caption{The plot presents 5-day holding period empirical exception rate 
densities when market portfolio data points are replaced with normal i.i.d. sample (top plot) and student-$t$ i.i.d. sample for $\nu=6$ (bottom plot).}\label{F:empirical4}
\end{figure}

\section*{Disclaimer and acknowledgements}
The views and opinions expressed in this article are the authors' own and do not necessarily reflect the views and opinions of their current or past employers. The authors report no conflicts of interest. The authors alone are responsible for the content and writing of the paper.

Marcin Pitera and \L{}ukasz Stettner acknowledge support from the National Science Centre, Poland, via project 2020/37/B/HS4/00120.

 {\small
 \bibliographystyle{agsm}
 \bibliography{bibliography}
 }

 \end{document}

%% file: headings.tex
\usepackage{amsthm, amsmath, amssymb, amsfonts, graphicx, epsfig}
\usepackage{accents}

\usepackage{bbm}
\usepackage{mathrsfs,dsfont}
\usepackage{multirow}
\usepackage{booktabs,nicefrac}
\usepackage{array}
\usepackage{caption}
\usepackage[normalem]{ulem}
\usepackage{cancel}

\usepackage{epstopdf}

\usepackage{graphicx}
\usepackage[font=sl,labelfont=bf]{caption}
\usepackage{subcaption}

\usepackage{float}
\restylefloat{table}

\usepackage{enumerate}
\usepackage[colorlinks=true, pdfstartview=FitV, linkcolor=blue,
            citecolor=blue, urlcolor=blue]{hyperref}
\usepackage[usenames]{color} 

\usepackage[round]{natbib}

\usepackage{url}
\makeatletter\def\url@leostyle{%
 \@ifundefined{selectfont}{\def\UrlFont{\sf}}{\def\UrlFont{\scriptsize\ttfamily}}} \makeatother

\newtheorem{theorem}{Theorem}
\newtheorem{proposition}[theorem]{Proposition}

\theoremstyle{definition}
\newtheorem{definition}[theorem]{Definition}
\newtheorem{example}[theorem]{Example}
\theoremstyle{remark}
\newtheorem{remark}[theorem]{Remark}

\numberwithin{equation}{section}
\numberwithin{theorem}{section}

\definecolor{Red}{rgb}{0.8,0,0.1}

\def\cF{\mathcal{F}}

\def\cN{\mathcal{N}}

\def\cZ{\mathcal{Z}}

\def\bE{\mathbb{E}}
\def\bF{\mathbb{F}}

\def\bN{\mathbb{N}}

\def\bR{\mathbb{R}}

\newcommand{\1}{\mathbbm{1}}            

\DeclareMathOperator*{\argmin}{arg\,min} 

\DeclareMathOperator{\var}{\mathrm{V}@\mathrm{R}}           

\makeatletter
\def\namedlabel#1#2{\begingroup
    #2%
    \def\@currentlabel{#2}%
    \phantomsection\label{#1}\endgroup
}
\makeatother